\documentclass[a4paper,reqno]{amsart}
\usepackage{amsmath}
  \usepackage{paralist}
 \usepackage[colorlinks=true]{hyperref}
\hypersetup{urlcolor=blue, citecolor=red}

\usepackage{amssymb}
\usepackage{amstext}
\usepackage{amsthm}
\usepackage{amsfonts}
\usepackage{amsxtra}
\usepackage{graphicx}
\usepackage{a4wide}

\newcommand{\N}{\mathbb N}

\newcommand{\R}{\mathbb R}
\newcommand{\C}{\mathbb C}

\newcommand{\de}{\partial}

\newcommand{\ha}{\mathcal H}

\newcommand{\norm}[1]{\left\Vert #1\right\Vert}
\newcommand{\vc}[1]{\boldsymbol #1}

\newcommand{\lap}{\Delta}

\newcommand{\dvg}{\operatorname{\mathrm{div}}}
\newcommand{\sym}{\operatorname{\mathrm{Sym}}}

\newcommand{\vt}[1]{\mathsf #1}

\newcommand{\bs}{\Sigma}
\newcommand{\rey}{\mathit{Re}}

\newcommand{\tsp}{\intercal}

\newcommand{\lop}{\operatorname{\mathcal{L}}}

\newtheorem{thm}{Theorem}[section]
\newtheorem{lem}[thm]{Lemma}

\theoremstyle{definition}

\theoremstyle{definition}
\newtheorem{rem}{Remark}
\theoremstyle{remark}

\title[1D rigid bodies in hyperviscous fluids] 
{Nonlinear free fall of one-dimensional rigid bodies in hyperviscous fluids}

\author[G. G. Giusteri, A. Marzocchi and A. Musesti]{}


 \email{giulio.giusteri@unicatt.it}
 \email{alfredo.marzocchi@unicatt.it}
 \email{alessandro.musesti@unicatt.it}

\begin{document}

\maketitle

\centerline{\scshape Giulio G. Giusteri, Alfredo Marzocchi and
  Alessandro Musesti }
\medskip
{\footnotesize
 \centerline{Dipartimento di Matematica e Fisica ``N. Tartaglia''}
   \centerline{Universit\`a Cattolica del Sacro Cuore}
   \centerline{Via dei Musei 41, I-25121 Brescia, Italy}
}
\bigskip


\begin{abstract}
We consider the free fall of slender rigid bodies in a viscous incompressible fluid. We show that the dimensional reduction (DR), performed by substituting the slender bodies with one-dimensional rigid objects, together with a hyperviscous regularization (HR) of the Navier--Stokes equation for the three-dimensional fluid lead to a well-posed fluid-structure interaction problem. In contrast to what can be achieved within a classical framework, the hyperviscous term permits a sound definition of the viscous force acting on the one-dimensional immersed body, and global-in-time existence and uniqueness of a solution can be proved. Those results show that the DR/HR procedure can be effectively employed for the mathematical modeling of the free fall problem in the slender-body limit. 
\end{abstract}

\section{Introduction}

Composite systems where microscopic bodies move in a viscous fluid are ubiquitous in both biological and technological contexts. The shape of such immersed bodies often displays some slenderness: one can encounter flat bodies, elongated bodies, somehow entangled filaments, and often simply point-like particles. For both mechanical and computational reasons it is tempting to approximate slender three-dimensional bodies with lower-dimensional objects,
and a slender-body theory, especially for low-Reynolds-number flows, has been widely developed (see \emph{e.g.}\ \cite{Bat70,ChwWu75,Lig75,Lig76,KelRub76,Joh80}).
Nevertheless, such a theory often relies on asymptotic expansions and sometimes crude approximations, mainly because of the difficulty in dealing with complex geometries and of the lack of a sound mathematical model for the interaction between three-dimensional fluids and lower-dimensional immersed bodies. The key obstacle to the development of such models is discussed in Remark~\ref{rem:nonex}.

Here we follow the spirit of a new slender-body theory proposed in~\cite{GiuFri12}, which, in contrast to previously developed alternatives, might be viewed as a modeling scheme, based on the concepts of dimensional reduction and hyperviscous regularization. The geometry of slender bodies immersed in a viscous fluid is approximated by lower-dimensional objects, while adding a hyperviscous term to the flow equation. The hyperviscosity, an additional parameter entering the equation, is given by the product of the ordinary viscosity with the square of a length that replaces the characteristic size of the body along the dimensions that vanish in the slender-body limit.

The classical Navier--Stokes equation for incompressible Newtonian fluids reads
\begin{equation*}
\rho\frac{\de\vc u}{\de t}+\rho(\vc u\cdot\nabla)\vc u+\nabla p-\mu\lap\vc u=\rho\vc b\,,
\end{equation*}
where $p$ is the pressure field, $\vc u$ is the divergence-free velocity field, $\rho>0$ is the constant and homogeneous mass density, $\mu>0$ is the dynamic viscosity and $\rho\vc b$ is a volumetric force density. It is well-known that the regularity of solutions for the three-dimensional Navier--Stokes equation is still an open problem, and various modifications of that equation with better regularity theories have been analyzed. Among those, the hyperviscous regularization (see, for example, \cite{Lio69}, Chap.~I, Remarque 6.11) consists of adding a term proportional to $\lap\lap\vc u$ to the equation. For this modified equation,
\begin{equation*}
\rho\frac{\de\vc u}{\de t}+\rho(\vc u\cdot\nabla)\vc u+\nabla p-\mu\lap\vc u
+\zeta\lap\lap\vc u=\rho\vc b\,,
\end{equation*}
where $\zeta>0$ is the hyperviscosity, existence and uniqueness of regular solutions have been established.

Despite the mathematical appeal of the hyperviscous regularization, it is generally hard to assign a relevant physical meaning to $\zeta$, other than that of a higher-order dissipation coefficient. In a series of papers \cite{FriGur06,Mus09,GiuMar11,Giu12}, different contributions to $\zeta$ associated with dissipation functionals are introduced and analyzed.
Here, as stated above, we assign to $\zeta$ a geometric, rather than dynamical, meaning, by introducing the \emph{effective thickness} $L>0$ of the lower-dimensional objects, and setting $\zeta=\mu L^2$, so that the hyperviscous flow equation becomes 
\begin{equation}\label{eq:HNS}
\rho\frac{\de\vc u}{\de t}+\rho(\vc u\cdot\nabla)\vc u+\nabla p-\mu\lap\left(\vc u
-L^2\lap\vc u\right)=\rho\vc b\,.
\end{equation}

It has been shown in~\cite{GiuMar10} that the flow generated by one-dimensional rigid bodies, moving with a \emph{prescribed} velocity in a fluid governed by the hyperviscous flow equation, can be uniquely determined. The purpose of the present note is to establish the existence and uniqueness of solutions for the \emph{free fall} of a slender rigid body $\bs$ in a hyperviscous fluid. In the problem at hand, the velocity of the body is no longer prescribed, and it has to be determined, together with the fluid velocity, by suitably coupling the flow equations with the equations describing the motion of the rigid body.
The free fall of three-dimensional rigid bodies in a Newtonian fluid has been widely studied; an excellent review of this topic can be found in~\cite{Gal02}, and we will follow similar arguments in our analysis. The main novelty here is, of course, the gap between the dimensions of the body and those of the fluid, which can be handled by a suitable definition, presented in Section~\ref{sec:drhr}, of the hydrodynamic force exerted by the hyperviscous fluid on the rigid body.

\section{Formulation of the free fall problem}\label{sec:formulation}

The free fall problem is characterized by the fact that the rigid body is immersed and dropped from rest in an otherwise quiescent fluid filling all of space, and gravity is the only external force acting on the system.
We represent a rigid body as a connected and bounded closed subset $\bs$ of $\R^3$ which is a finite union of images of $[0,1]$ through a $C^1$-diffeomorphism. In this way, the one-dimensional Hausdorff measure of $\bs$, $\ha^1(\bs)$, is finite and positive.
It is convenient to write the problem in a co-moving frame, with origin at the center of mass $\vc c(t)$ of $\bs$.  
Denoting by $\vc y$ the position of a point in the original inertial frame, and by $\vc x$ its position in the co-moving frame, we know that, at any time $t\geq 0$,
\[
\vc x=\vt Q^\tsp(t)(\vc y-\vc c(t))\,,
\]
where $\vt Q(t)$ is an orthogonal linear transformation for any $t\geq 0$, with $\vt Q(0)=\vt 1$.
If the velocity of the center of mass and the spin of the rigid body in the inertial frame are denoted by $\vc\eta(t)$ and $\vc\varOmega(t)$, respectively, so that
\[
\vc v(t)=\vc\eta(t)+\vc\varOmega(t)\times(\bar{\vc y}-\vc c(t))
\]
is the velocity, in that frame, of any point $\bar{\vc y}$ belonging to the rigid body, then their expression in the co-moving frame is given by
\[
\vc\xi(t):=\vt Q^\tsp(t)\vc\eta(t)\quad\text{and}\quad\vc\omega(t):=\vt Q^\tsp(t)\vc\varOmega(t)\,,
\]
respectively, and the rigid velocity field $\vc v$ is transformed into
\[
\vc U(t):=\vc\xi(t)+\vc\omega(t)\times\bar{\vc x}\,, 
\]
where $\bar{\vc x}$ denotes the coordinates of $\bar{\vc y}$ in the co-moving frame.

We also note that the gravitational acceleration vector $\vc g$ (constant in the inertial frame) is represented, in the co-moving frame, by $\vc G(t):=\vt Q^\tsp(t)\vc g$, which is easily seen to satisfy the ordinary differential equation
\begin{equation}\label{eq:G}
\frac{d\vc G}{dt}=\vc G\times\vc\omega\,.
\end{equation}

As customary when studying flows past rigid bodies, the velocity field $\vc u$ that we consider is the so called \emph{disturbance field}, which is the difference between the actual flow and the flow at infinity, both seen in the co-moving frame. Since the flow at infinity is $-\vc U$ (that is minus the extension to all of the fluid of the motion of the immersed object), the representation of the fluid flow in the co-moving frame is given by $\vc u-\vc U$.

The continuity and flow equations for an incompressible (disturbance) velocity field $\vc u(\vc x,t)$ and pressure field $p(\vc x,t)$ defined on $(\R^3\setminus\bs)\times[0,+\infty)$, become
\begin{equation}\label{eq:incomp}
\dvg\vc u=0\,,
\end{equation}
\begin{equation}\label{eq:floweq1}
\rho\left(\frac{\de\vc u}{\de t}+[(\vc u-\vc U)\cdot\nabla]\vc u+\vc\omega\times\vc u\right)=\dvg\vt T(\vc u,p)+\rho\vc G\,,
\end{equation}
where $\vt T(\vc u,p)$ denotes the Cauchy stress tensor.
Note that, thanks to its frame indifference properties, $\vt T$ retains the same functional dependence on the velocity field seen in both the inertial frame and in the co-moving one.

We have also a decay condition
\begin{equation}\label{eq:decay}
\lim_{|\vc x|\to\infty}\vc u(\vc x,t)=\vc 0\,,
\end{equation}
and the adherence to the rigid body, given by
\begin{equation}\label{eq:adherence}
\vc u(\vc x,t)=\vc U(\vc x,t)\qquad\text{on }\bs\times[0,+\infty)\,.
\end{equation}

The equations of motion for the rigid body in the co-moving frame are then
\begin{equation}\label{eq:linmom}
m\frac{d\vc\xi}{dt}+m\vc\omega\times\vc\xi=m_e\vc G+\vc f
\end{equation}
and
\begin{equation}\label{eq:angmom}
\vt J\frac{d\vc\omega}{dt}+\vc\omega\times(\vt J\vc\omega)=-m_c(\vc r\times\vc G)+\vc t\,,
\end{equation}
where $m$ is the total mass of the rigid body, $m_e$ is its \emph{effective mass}, $m_c$ is the \emph{complementary mass}, $\vt J$ its inertia tensor, $\vc r$ is the position of the centroid\footnote{The centroid of a rigid body coincide with its center of mass when the body has a uniform mass density; in the latter case $\vc r=\vc 0$.} of $\bs$ in the co-moving frame, and $\vc f$ and $\vc t$ are the total hydrodynamic force and torque, respectively, exerted on the rigid body as a consequence of the fluid flow $\vc u$ and pressure $p$. The proper definition of $\vc f$ and $\vc t$ is discussed in Section~\ref{sec:drhr}. In view of the interpretation of one-dimensional bodies as representations for real three-dimensional objects,
to properly account for Archimedean forces which would vanish in the slender-body limit, we have to consider from the very beginning $m_e=m-m_c$ as the difference between the real mass of the object and the mass $m_c$ of a portion of fluid occupying the real volume of the object, and we must add the terms proportional to $m_e$ and $m_c$ in equations~\eqref{eq:linmom} and~\eqref{eq:angmom}, respectively.

The whole set of equations \eqref{eq:G}--\eqref{eq:angmom} represents the differential problem associated with the free fall of a rigid object $\bs$ in an incompressible fluid. The physical properties of such a fluid are encoded in the Cauchy stress tensor $\vt T$, and the peculiar form we are going to use throughout the present paper is discussed in Section~\ref{sec:drhr}.

It is convenient to consider the non-dimensional form of problem \eqref{eq:G}--\eqref{eq:angmom}: by choosing suitable reference length
$d$, proportional to the diameter of $\bs$, and speed $W=\rho
  g d^2/\mu$, we can switch to non-dimensional quantities
according to
\begin{gather*}
\vc x\to\frac{\vc x}{d}\;,\quad t\to\frac{t\mu}{\rho d^2}\;,\quad\vc
u\to\frac{\vc u}{W}\;,\quad \vc \xi\to\frac{\vc \xi}{W}\;,\quad
\vc\omega\to\frac{\vc \omega d}{W}\;,\\
p\to\frac{pd}{\mu W}\;,\quad
m\to\frac{m}{\rho d^3}\;,\quad\vc G\to \frac{\vc G}{g}\;,\quad\vc g\to \frac{\vc g}{g}\;,
\end{gather*}
obtaining
\begin{align}
&\dvg\vc u=0\,,\label{eq:nd1}\\
&\frac{\de\vc u}{\de t}+\rey\left\{[(\vc u-\vc U)\cdot\nabla]\vc u+\vc\omega\times\vc u\right\}=\dvg\vt T(\vc u,p)+\vc G\,,\label{eq:nd2}\\
&\lim_{|\vc x|\to\infty}\vc u(\vc x,t)=\vc 0\,,\\
&\vc u(\vc x,t)=\vc U(\vc x,t)\qquad\text{on }\bs\times[0,+\infty)\,,\label{eq:nd4}\\
&m\frac{d\vc\xi}{dt}+\rey(m\vc\omega\times\vc\xi)=m_e\vc G+\vc f\,,\label{eq:nd5}\\
&\vt J\frac{d\vc\omega}{dt}+\rey[\vc\omega\times(\vt J\vc\omega)]=-m_c(\vc r\times\vc G)+\vc t\,,\label{eq:nd6}\\
&\frac{d\vc G}{dt}=\rey(\vc G\times\vc\omega)\,,\label{eq:nd7}
\end{align}
with initial conditions
\begin{equation}\label{eq:nd8}
\vc u(\vc x,0)=\vc\xi(0)=\vc\omega(0)=\vc 0\,,\quad\vc G(0)=\vc g\,.
\end{equation}
where $\rey=\rho Wd/\mu=\rho^2gd^3/\mu^2$ is the Reynolds number and every quantity has
to be understood as non-dimensional.

The low-Reynolds-number approximation of the differential problem,
which is also a linearization of the equation for the flow, is
obtained by neglecting the terms proportional to $\rey$ in
equations~\eqref{eq:nd2}, \eqref{eq:nd5}, and \eqref{eq:nd6}. When
considering free fall problems, the energy budget is determined by
gravitational forces and viscous dissipation; hence the limit $\rey\to
0$ corresponds to the situation where the latter prevails. In the
meanwhile, the geometric parameters $d$ and $L$ do not need to be small,
even in that limit.
Notice that equation \eqref{eq:nd7} remains unchanged, since it
represents a geometric constraint which holds for any non-vanishing
value of $\rey$.

\section{The viscous force acting on a slender body}\label{sec:drhr}

The constitutive theory for non-simple fluids leading to a hyperviscous flow equation has been developed in~\cite{FriGur06,Mus09,GiuMar11}. It offers a number of possible choices for the terms to be included in $\vt T$, in addition to those of Newtonian fluids. Here we make a somewhat minimal assumption, obtaining a fluid which is quasi-Newtonian, while being able to adhere to lower-dimensional objects. As explained in~\cite{Giu12}, we can simply take
\begin{equation}\label{eq:defT}
\vt T(\vc u,p):=-p\vt 1+\left(\nabla\vc u+\nabla\vc u^\tsp-\ell^2\nabla\lap\vc u\right)\,,
\end{equation}
where we used the non-dimensional effective thickness $\ell=L/d$.
That choice for $\vt T$ has also the feature of introducing only one new parameter, $\ell$, to which we have already assigned a geometric meaning. It is straightforward to check that $\vt T$, defined as in~\eqref{eq:defT}, enjoys the standard symmetry and frame indifference properties.

Let $C$ be any closed bounded subset of $\R^3$ and $r>0$. The set
\[
V_r(C):=\left\{\vc x\in\R^3 : d(\vc x,C)\leq r \right\}\,,
\]
where $d(\vc x,C)$ denotes the distance of $\vc x$ from $C$, is the $r$-neighborhood of $C$. We define the total hydrodynamic force, due to the fluid velocity and pressure field $(\vc u,p)$, acting on the slender body $\bs$ as
\begin{equation}\label{eq:defforce}
\vc f(\vc u,p):=\lim_{r\to 0}\int_{\de V_r(\bs)}\vt T(\vc u,p)\vc n\,,
\end{equation}
where $\vc n$ denotes the unit outer normal to $\de V_r(\bs)$, and $\vt T$ is the stress tensor defined in~\eqref{eq:defT}. Notice that, thanks to the regularity of $\bs$, there exists always $\bar r>0$ small enough so that $V_r(\bs)$ has a Lipschitz boundary for any $r\leq\bar r$.

Let us check that the limit in~\eqref{eq:defforce} is indeed well-defined. We consider a ball $B_R$ centered at the origin and with very large radius $R$, which contains $\bs$.
According to equation~\eqref{eq:HNS}, the term $\dvg\vt T$ is required to balance both the inertial term and the external forces acting on the system. In the free fall problem considered in this paper, the only external interaction is gravity and, due to the dimensional reduction we perform on the rigid body, it is represented by a measure whose singular part is concentrated on $\bs$. We can see this by noting that the mass density per unit volume must diverge on $\bs$ to give a non-zero weight to a body with vanishing volume.

Hence we expect also $\dvg\vt T$ to be a measure and, denoting by $\alpha$ ($\varsigma$) its absolutely continuous (singular) part with respect to the three-dimensional Lebesgue measure, we expect the support of $\varsigma$ to be contained in $\bs$ and that $\dvg\vt T=\alpha$ in $\R^3\setminus\bs$. Under these assumption we have, for any $r>0$ small enough,
\begin{equation*}
\lim_{r\to 0}\int_{B_R\setminus V_r(\bs)}\dvg\vt T=\lim_{r\to 0}\int_{B_R\setminus V_r(\bs)}\alpha=\int_{B_R\setminus\bs}\alpha=\int_{B_R\setminus\bs}\dvg\vt T\,,
\end{equation*}
by applying Lebesgue's theorem, since $\alpha\in L^1(\R^3)$; then, by the Divergence theorem,
\begin{equation}\label{eq:limf}
\vc f=\lim_{r\to 0}\int_{\de V_r(\bs)}\vt T\vc n=\int_{\de B_R}\vt T\vc n-\int_{B_R\setminus\bs}\dvg\vt T\,,
\end{equation}
where $\vc n$ is always the outer normal. Since the right-hand side of~\eqref{eq:limf} is independent of $r$, the left-hand side is well defined. It is important to stress the fact that, if $\varsigma$ were absent, then the integral over $B_R\setminus\bs$ of $\dvg\vt T$ would be equal to the integral over all of $B_R$ and $\vc f$ would simply vanish. 

In a similar fashion, we define the total hydrodynamic torque acting on $\bs$, due to the fluid velocity and pressure field $(\vc u,p)$, as
\begin{equation}\label{eq:deftorque}
\vc t(\vc u,p):=\lim_{r\to 0}\int_{\de V_r(\bs)}\vc x\times\vt T(\vc u,p)\vc n\,.
\end{equation}


\section{Steady free fall}

In what follows we prove the existence of a solution for the steady version of the differential problem introduced in section~\ref{sec:formulation}:
\begin{align}
&\dvg\vc u=0\,,\label{eq:nl1}\\
&\rey\left\{[(\vc u-\vc U)\cdot\nabla]\vc u+\vc\omega\times\vc u\right\}=\dvg\vt T(\vc u,p)+\vc g\qquad\text{in }\R^3\setminus\bs\,,\label{eq:nl2}\\
&\lim_{|\vc x|\to\infty}\vc u(\vc x)=\vc 0\,,\\
&\vc u(\vc x)=\vc\xi+\vc\omega\times\vc x\qquad\text{on }\bs\,,\label{eq:nl4}\\
&\rey(m\vc\omega\times\vc\xi)=m_e\vc g+\vc f\,,\label{eq:nl5}\\
&\rey[\vc\omega\times(\vt J\vc\omega)]=-m_c(\vc r\times\vc G)+\vc t\,,\label{eq:nl6}\\
&\vc g\times\vc\omega=\vc 0\,.\label{eq:nl7}
\end{align}
We will first treat the low-Reynolds-number version of problem \eqref{eq:nl1}--\eqref{eq:nl7}, thus reducing the nonlinearity to the sole term $\vc g\times\vc\omega$. The linear nature of the flow equations in the low-Reynolds-number limit allows for a finer analysis of the relation between the shape of the objects and their steady motion, which is developed in \cite{GiuMar13}. Here we only present an existence result, studying then the full nonlinear problem, which is the main topic of the present paper. 

\subsection{The low-Reynolds-number limit}

Now we briefly discuss the functional setting. In view of the natural
variational formulation of the problem, we introduce the space
\[
C:=\{\vc u\in\C^\infty_0(\R^3;\R^3):\ \dvg\vc u=0\}
\]
endowed with the norm
\[
\|\vc u\|^2:= \int_{\R^3}\left(2|\sym\nabla\vc
  u|^2+\ell^2|\lap\vc u|^2\right),
\]
where $\sym\nabla\vc u:=(\nabla\vc u+\nabla\vc u^\tsp)/2$.
Denote with $X$ the completion of $C$ in that norm; it is easy to
see that if $\vc u\in X$, then $\nabla\vc u$ belongs to the Sobolev
space $W^{1,2}(\R^3;\R^9)$. Moreover, by Sobolev and Morrey's
Theorems, $X$ embeds in $L^6(\R^3;\R^3)$ and also in a space of
H\"older-continuous functions.  Regarding the pressure $p$ as the
Lagrange multiplier of the constraint $\dvg\vc u=0$, we will take it
in the dual Sobolev space $W^{-1,2}(\R^3)$.

We summarize the problem of the steady free fall of a one-dimensional
body $\bs$ at low Reynolds number, as the following: find $(\vc
u,p)\in X\times W^{-1,2}(\R^3)$ and
$\vc\xi,\vc\omega,\vc g\in\R^3$ with $|\vc g|=1$, such that
\begin{align}
& \nabla p + \lop\vc u = \vc g
\quad\text{on $\R^3\setminus \bs$},\label{eq:st2}\\
&\vc u(\vc x)=\vc\xi+\vc\omega\times\vc x\quad\text{on $\bs$},\label{eq:st3}\\
&m_e\vc g = -\vc f,\label{eq:st4}\\
&m_c\vc r\times\vc g = \vc t,\label{eq:st5}\\
&\vc g\times\vc\omega=\vc 0,\label{eq:st6}
\end{align}
where we set $\lop\vc u:=-\lap\vc u+\ell^2\lap\lap\vc u$. The constraint $\dvg\vc u=0$ is encoded in the definition of the space
$X$, while the strong decay condition~\eqref{eq:decay} is replaced by
an integrability condition for $\vc u$ on the whole $\R^3$. Notice
that, although equation~\eqref{eq:st2} is linear, the full problem is
nonlinear.

To prove the existence of a solution of the steady free fall problem we first
introduce some auxiliary problems, which are well-posed by virtue of
the following result.

\begin{lem}\label{lem:aux}
Given $\vc\xi,\vc\omega\in\R^3$, there exists a unique
solution $(\vc h,p)\in X\times W^{-1,2}(\R^3)$ of the problem
\begin{equation}\label{eq:aux}
\begin{cases}
\nabla p- \Delta \vc h +\ell^2\Delta\Delta \vc h=0 &\text{in
  $\R^3\setminus \bs$,}\\
\vc h=\vc\xi+\vc\omega\times\vc x &\text{on $\bs$.}
\end{cases}
\end{equation}
Moreover $\vc h\in W^{3,q}_{loc}(\R^3;\R^3)$ for any $1<q<\frac 3 2$.
\end{lem}
\begin{proof}
  Since $X$ embeds in a space of H\"older-continuous functions, the
  subset
\[
\{\vc v\in X:\ \vc v=\vc\xi+\vc\omega\times\vc x\text{ on }\bs\}
\]
is well-defined, closed and convex. The velocity field $\vc h$ can be
found by minimizing on that set the functional
\[
\mathcal F(\vc v):=\frac{1}{2}\|\vc v\|^2=\frac 1 2 \int_{\R^3}\left(2|\sym\nabla\vc
  v|^2+\ell^2|\lap\vc v|^2\right).
\]
Being $\mathcal F$ a strictly convex functional, $\vc h$ is
unique. Then, the pressure field $p$ can be recovered as the Lagrange
multiplier of the divergence-free constraint.

Since the adherence condition on $\bs$ can be replaced by a non
homogeneous right-hand side which is a measure supported on $\bs$,
that is
\[
\nabla p- \Delta \vc h +\ell^2\Delta\Delta \vc h=\vc \eta \quad\text{in $\R^3$,}\\
\]
where $\vc\eta$ is a (vector-valued) Radon measure which vanishes
outside $\bs$, we get for $\vc h$ a fourth-order linear elliptic
equation with a measure-valued datum. Since the space of
Radon measures embeds in $W^{-1,q}_{loc}(\R^3;\R^3)$ for
every $1<q<\frac 3 2$ (here $\frac 3 2$ is such that $q'>n$ in the
case $n=3$), then a standard regularity gain of the solution
\cite[Theorem $15.3'$]{AgmDou59} entails $\vc h\in
W^{3,q}_{loc}(\R^3;\R^3)$.
\end{proof}

\begin{rem}\label{rem:nonex}
In view of the proof of Lemma~\ref{lem:aux} it is clear that, for $\vc h$ to truly satisfy the adherence condition on $\bs$, the latter must have \emph{positive} $H^2$-capacity. That property is granted by the embedding of $X$ into a space fo H\"older continuous functions.
On the contrary, when treating the classical Stokes equation, it is natural to require only that the gradient of the velocity field is in $L^2(\R^3;\R^9)$. Now, the $H^1$-capacity of $\bs$ vanishes, allowing us to arbitrarily change the value of the velocity field on $\bs$, so that the imposition of an adherence condition does not really affect the flow.
Indeed it is immediate to see that, given $\vc\xi\in\R^3$ and $\vc\omega\in\R^3$, the solution $\vc u\in H_{loc}^1(\R^3;\R^3)$ of the problem
\begin{equation*}
\left\{\begin{aligned}
&\dvg\vc u=0 &\text{in }\R^3\,,\\
&\nabla p-\mu\lap\vc u=0 &\text{in }\R^3\,,\\
&\vc u=\vc\xi+\vc\omega\times\vc x &\text{on }\bs\,\,,
\end{aligned}\right.
\end{equation*}
is such that
\[
\int_{\R^3}\vert\vc u\vert^2=0\,,
\]
and hence $\vc u=0$ almost everywhere in $\R^3$.
\end{rem}

Consider now the solutions $(\vc h^{(i)},p^{(i)})$ and $(\vc
H^{(i)},P^{(i)})$ ($i=1,2,3$) in the space $X\times W^{-1,2}(\R^3)$ of
the auxiliary problems
\begin{equation}
\label{eq:aux1}
\begin{cases}
\nabla p^{(i)}- \Delta \vc h^{(i)} +\ell^2\Delta\Delta\vc h^{(i)}=0 &\text{in }\R^3\setminus\bs,\\
\vc h^{(i)}=\vc e_i &\text{on }\bs,
\end{cases}
\end{equation}
and
\begin{equation}
\label{eq:aux2}
\begin{cases}
\nabla P^{(i)}- \Delta \vc H^{(i)} +\ell^2\Delta\Delta \vc H^{(i)}=0 &\text{in }\R^3\setminus\bs,\\
\vc H^{(i)}=\vc e_i\times\vc x &\text{on }\bs.
\end{cases}
\end{equation}
We will show that the combinations
\begin{equation}\label{eq:combination}
\vc u=\sum_{i=1}^3 [\xi_i\vc h^{(i)}+\omega_i\vc H^{(i)}],\quad 
p=\sum_{i=1}^3 [\xi_i p^{(i)}+\omega_i P^{(i)}]+\vc g\cdot\vc x,
\end{equation}
for a suitable choice of the vectors $\vc\xi$ and $\vc\omega$, solve
the steady free fall problem.

\begin{thm}
The differential problem \eqref{eq:st2}--\eqref{eq:st6} admits a solution $(\vc u,p,\vc\xi,\vc\omega,\vc g)$.
\end{thm}
\begin{proof}
It is straightforward to check that the fields $\vc u$ and $p$ defined by~\eqref{eq:combination} satisfy equations \eqref{eq:st2} and \eqref{eq:st3}. Equation \eqref{eq:st6} implies that $\vc\omega=\lambda\vc g$ for some $\lambda\in\R$, and equations \eqref{eq:st4} and \eqref{eq:st5} reduce to the following algebraic system in the six scalar unknowns $\vc\xi$, $\lambda$, and $\vc g$:
\begin{equation}\label{eq:algsys}
\begin{aligned}
\vt K\vc\xi+\lambda\vt S\vc g&\mbox{}=m_e\vc g\,,\\
\vt C\vc\xi+\lambda\vt B\vc g&\mbox{}=-m_c(\vc r\times\vc g)\,,
\end{aligned}
\end{equation}
where the matrices $\vt K$, $\vt S$, $\vt C$, and $\vt B$ are defined by
\[
\vt K_{ji}:=-\lim_{r\to 0}\int_{\de V_r(\bs)}\vt T(\vc h^{(i)},p^{(i)})\vc n\cdot\vc e_j\,,
\]
\[
\vt S_{ji}:=-\lim_{r\to 0}\int_{\de V_r(\bs)}\vt T(\vc H^{(j)},P^{(j)})\vc n\cdot\vc e_i\,,
\]
\[
\vt C_{ji}:=-\lim_{r\to 0}\int_{\de V_r(\bs)}\vc x\times\vt T(\vc h^{(j)},p^{(j)})\vc n\cdot\vc e_i\,,
\]
\[
\vt B_{ji}:=-\lim_{r\to 0}\int_{\de V_r(\bs)}\vc x\times\vt T(\vc H^{(i)},P^{(i)})\vc n\cdot\vc e_j\,.
\]

It is now clear that the steady free fall problem admits a solution if and only if the algebraic system~\eqref{eq:algsys} admits a solution, and the latter fact is related to the properties of the matrix
\[
\vt A:=
\begin{pmatrix}
\vt K & \vt S \\
\vt C & \vt B
\end{pmatrix}\,.
\]
By the linearity of~\eqref{eq:aux}, the instantaneous energy dissipation rate $E$ of the flow which takes the value $\vc\xi+\vc\omega\times\vc x$ on $\bs$ is given by (see also \cite{HapBre65})
\[
E=
\begin{pmatrix}
\vt K & \vt S \\
\vt C & \vt B
\end{pmatrix}
\begin{pmatrix}
\vc\xi \\ \vc\omega 
\end{pmatrix}\cdot
\begin{pmatrix}
\vc\xi \\ \vc\omega 
\end{pmatrix}\,.
\]
Now, the constitutive prescription for the Cauchy stress tensor has been chosen to be thermodynamically consistent, thus implying that $E$ be positive for any non-vanishing flow; hence, the matrices $\vt A$, $\vt K$, and $\vt B$ turn out to be positive definite and, in particular, invertible.

We can now transform~\eqref{eq:algsys} into
\begin{gather}\label{eq:finsys}
\vc\xi\mbox{}=\vt K^{-1}(m_e\vc g-\lambda\vt S\vc g)\,,\\
\vt F\vc g:=(\vt C\vt K^{-1}\vt S-\vt B)^{-1}(m_e\vt C\vt K^{-1}\vc g+m_c\vc r\times\vc g)\mbox{}=\lambda\vc g\,.
\end{gather}
Since $\vt A$ is non-singular, the matrix $\vt F$ is a well-defined $3\times 3$ real matrix and it has at least one real eigenvalue. Such an eigenvalue $\lambda$, the associated unit eigenvector $\vc g$ and $\vc\xi$ calculated as in \eqref{eq:finsys}, together with the fields $\vc u$ and $p$ introduced in~\eqref{eq:combination}, furnish a solution for equations \eqref{eq:st2}--\eqref{eq:st6}.
\end{proof}

\subsection{Nonlinear free fall}

To solve the differential problem \eqref{eq:nl1}--\eqref{eq:nl7} we employ a weak formulation. We first consider the space
\begin{equation*}
V:=\left\{\vc v\in C_0^\infty(\R^3;\R^3):\dvg\vc v=0\text{ and }\vc v(\vc x)=\vc\eta+\vc\psi\times\vc x\text{ in }V_r(\bs),\;\exists\;r>0  \right\}\,,
\end{equation*}
where $\vc\eta,\vc\psi\in\R^3$ represent the virtual vectors associated with $\vc\xi$ and $\vc\omega$, respectively. Now, denoting by $\norm{\cdot}_2$ the norm in $L^2(\R^3;\R^3)$, the space $Y$ is the completion of $V$ with respect to the norm
\begin{equation*}
\norm{\vc v}^2_Y:=2\norm{\sym\nabla\vc v}_2^2+\ell^2\norm{\lap\vc v}^2_2\,.
\end{equation*}
It is easy to show that, for any $\vc v\in Y$, we have $\int_{\R^3}\vc v=0$. Indeed, take any $\vc a\in\R^3$ and any $\vc v\in V$: in particular $\vc v$ has compact support, and there exists a ball $B_R$ of radius $R$ large enough to ensure that
\begin{equation*}
\vc a\cdot\int_{\R^3}\vc v=\int_{B_R}\vc a\cdot\vc v=\int_{B_R}\dvg[(\vc a\cdot\vc x)\vc v]=\int_{\de B_R}(\vc a\cdot\vc x)(\vc v\cdot\vc n)=0\,.
\end{equation*}
The result can be extended to all of $Y$ by a density argument.

We now give a formal derivation of the weak form of our differential problem.
We multiply \eqref{eq:nl2} by $\vc v\in V$ and integrate over $\R^3\setminus V_r(\bs)$ with $r>0$ sufficiently small, obtaining
\begin{multline}\label{eq:mult}
\rey\int_{\R^3\setminus V_r(\bs)}\left\{[(\vc u-\vc U)\cdot\nabla]\vc u+\vc\omega\times\vc u\right\}\cdot\vc v=
\int_{\R^3\setminus V_r(\bs)}\dvg\vt T(\vc u,p)\cdot\vc v+\int_{\R^3\setminus V_r(\bs)}\vc g\cdot\vc v\\
=-2\int_{\R^3\setminus V_r(\bs)}\sym\nabla\vc u\cdot\nabla\vc v-\ell^2\int_{\R^3\setminus V_r(\bs)}\lap\vc u\cdot\lap\vc v\\
-\int_{\de V_r(\bs)}\vt T(\vc u,p)\vc n\cdot\vc v-\ell^2\int_{\de V_r(\bs)}\lap\vc u\otimes\vc n\cdot\nabla\vc v+\vc g\cdot\int_{\R^3\setminus V_r(\bs)}\vc v\,.
\end{multline}
We have
\begin{equation*}
\int_{\de V_r(\bs)}\lap\vc u\otimes\vc n\cdot\nabla\vc v=\int_{V_r(\bs)}\dvg(\nabla\vc v^\tsp\lap\vc u)=\int_{V_r(\bs)}\nabla\lap\vc u\cdot\nabla\vc v\,,
\end{equation*}
and, letting $r\to 0$,
\begin{multline*}
\lim_{r\to 0}\int_{\de V_r(\bs)}\vt T(\vc u,p)\vc n\cdot\vc v
=\lim_{r\to 0}\left(\vc\eta\cdot\int_{\de V_r(\bs)}\vt T(\vc u,p)\vc n+\vc\psi\cdot\int_{\de V_r(\bs)}\vc x\times\vt T(\vc u,p)\vc n\right)\\=\vc\eta\cdot\vc f+\vc\psi\cdot\vc t\,.
\end{multline*}
Since we assume all the needed regularity on $\vc u$, we have $\nabla\lap\vc u\cdot\nabla\vc v\in L^1(\R^3;\R)$. The latter property guarantees that, by taking the limit $r\to 0$ in \eqref{eq:mult}, we obtain
\begin{multline*}
\rey\int_{\R^3}\left\{[(\vc u-\vc U)\cdot\nabla]\vc u+\vc\omega\times\vc u\right\}\cdot\vc v=\\
=-2\int_{\R^3}\sym\nabla\vc u\cdot\nabla\vc v-\ell^2\int_{\R^3}\lap\vc u\cdot\lap\vc v-\vc\eta\cdot\vc f-\vc\psi\cdot\vc t\,;
\end{multline*}
by substituting equations \eqref{eq:nl5}--\eqref{eq:nl6} into the preceding expression, we can write
\begin{multline}\label{eq:varform}
\rey\int_{\R^3}\left\{[(\vc u-\vc U)\cdot\nabla]\vc u+\vc\omega\times\vc u\right\}\cdot\vc v+\rey(m\vc\omega\times\vc\xi)\cdot\vc\eta+\rey[\vc\omega\times(\vt J\vc\omega)]\cdot\vc\psi=\\
=-2\int_{\R^3}\sym\nabla\vc u\cdot\nabla\vc v-\ell^2\int_{\R^3}\lap\vc u\cdot\lap\vc v+m_e\vc g\cdot\vc\eta-m_c(\vc r\times\vc g)\cdot\vc\psi\,,
\end{multline}
for any $\vc v\in V$, which represents the variational form of equations \eqref{eq:nl1}--\eqref{eq:nl6}.

While proving our main result we will make use of the following lemmas. The third one is proven in~\cite[Lemme 4.4]{Ser87}.

\begin{lem}\label{lem:rigid}
There exists a constant $K>0$, depending only on $\bs$, such that, for any $\vc v\in V$,
\begin{equation}
\label{eq:bound}
|\vc\eta|+|\vc\eta+\vc\psi\times\vc r|\leq K\norm{\vc v}_Y\,,  
\end{equation}
where $\vc\eta+\vc\psi\times\vc x$ is the value taken by $\vc v$ in a neighborhood of $\bs$.
\end{lem}
\begin{proof}
Let us first observe that Korn's inequality and standard embedding theorems provide a constant $K_1>0$ such that $\norm{\vc v}_{\infty}\leq K_1\norm{\vc v}_Y$ for any $\vc v\in V$. We then have
\begin{equation*}
\mathcal H^1(\bs)|\vc\eta+\vc\psi\times\vc r|=\left\vert\int_\bs \vc\eta+\vc\psi\times\vc x\,d\mathcal H^1(\vc x) \right\vert\leq\mathcal H^1(\bs) K_1\norm{\vc v}_Y\,,
\end{equation*}
and it is now enough to show that $|\vc\eta|\leq\norm{\vc v}_\infty$ to prove the assertion with $K=2K_1$. 

By noting that the velocity of the center of mass belongs to the convex hull in $\R^3$ of the set of velocities on $\bs$, we conclude that
\begin{equation}
\label{eq:speed}
|\vc\eta|\leq\sup_{\vc x\in\bs}|\vc\eta+\vc\psi\times\vc x|\leq\norm{\vc v}_\infty\,,
\end{equation}
and the proof is complete.
\end{proof}

\begin{lem}\label{lem:basis}
There exists $\{\vc\phi_i\}_{i\in\N}\subset V$ whose linear hull is dense in $Y$, and such that, for any $i,k\in\N$,
\begin{equation*}
2\int_{\R^3}\sym\nabla\vc\phi_i\cdot\nabla\vc\phi_k+\ell^2\int_{\R^3}\lap\vc\phi_i\cdot\lap\vc\phi_k=\delta_{ik}\,,
\end{equation*}
and, given $\vc v\in V$ and $\varepsilon>0$, there exist $m\in\N$ and $b_1,\ldots,b_m\in\R$ such that
\begin{equation}
\label{eq:approx}
\norm{\vc v-\sum_{i=1}^mb_i\vc\phi_i}_{C^2(\R^3;\R^3)}<\varepsilon\,.
\end{equation}
\end{lem}
\begin{proof}
Let $W$ be the completion of $V$ in the norm
\begin{equation*}
\norm{\vc v}_W:=\norm{\vc v}_Y+\sum_{|\alpha|=3}^4\norm{D^\alpha\vc v}_2\,,
\end{equation*}
where $|\alpha|$ denotes the length of the multindex $\alpha$. Clearly, $W$ is a separable subspace of $Y$, and $V$ is a separable subspace of $W$, hence there exists a $\{\vc\varphi_i\}_{i\in\N}\subset V$ dense in $W$. Since we are dealing with separable Hilbert spaces, we can obtain from $\{\vc\varphi_i\}_{i\in\N}$ a set $\{\vc\phi_i\}_{i\in\N}$ which is dense in $Y$ and orthonormal with respect to the scalar product on $Y$ which induces $\norm{\cdot}_Y$.
The approximation property \eqref{eq:approx} descends now from Korn's inequality and standard embedding theorems.
\end{proof}

\begin{lem}\label{lem:top}
Fix $R>0$ and let $B_R$ denote the ball in $\R^n$ of radius $R$ and center $\vc 0$. Let
\begin{equation*}
\begin{aligned}
&\mathcal T:\overline{B_R}\times S^2\to\R^n\,,\\
&\vc\tau:\overline{B_R}\times S^2\to\R^3
\end{aligned}
\end{equation*}
be continuous maps. If
\begin{equation*}
\forall\,(\vc c,\vc g)\in\de B_R\times S^2:\mathcal T(\vc c,\vc g)\cdot\vc c>0\,,\;
\text{ and}\quad\forall\,(\vc c,\vc g)\in B_R\times S^2:\vc\tau(\vc c,\vc g)\cdot\vc g=0\,,
\end{equation*}
then there exists $(\bar{\vc c},\bar{\vc g})\in B_R\times S^2$ such that
\begin{equation*}
\mathcal T(\bar{\vc c},\bar{\vc g})=\vc 0\quad\text{and}\quad\vc\tau(\bar{\vc c},\bar{\vc g})=\vc 0\,.
\end{equation*}
\end{lem}

\begin{thm}
The differential problem \eqref{eq:nl1}--\eqref{eq:nl7} admits a solution $(\vc u,p,\vc\xi,\vc\omega,\vc g)$.
\end{thm}
\begin{proof}
We consider, for any $n\in\N$, the truncated expansions on the basis furnished by Lemma \ref{lem:basis}
\begin{equation*}
\vc u_n:=\sum_{i=1}^nc_i\vc\phi_i\,,\quad \vc\xi_n:=\sum_{i=1}^nc_i\vc\eta_i\,,\quad \vc\omega_n:=\sum_{i=1}^nc_i\vc\psi_i\,,
\end{equation*}
where $\vc\phi_i=\vc\eta_i+\vc\psi_i\times\vc x$ on $\bs$,
and apply Galerkin's approximation to equation \eqref{eq:varform}.

For any $n\in\N$, we consider the maps defined by
\begin{equation*}
\vc\tau(\vc c,\vc g):=\vc g\times\vc\omega_n\,,
\end{equation*}
and
\begin{multline*}
[\mathcal T(\vc c,\vc g)]_k:=\rey\int_{\R^3}\left\{[(\vc u_n-\vc\xi_n-\vc\omega_n\times\vc x)\cdot\nabla]\vc u_n+\vc\omega_n\times\vc u_n\right\}\cdot\vc\phi_k\\ +\rey(m\vc\omega_n\times\vc\xi_n)\cdot\vc\eta_k+\rey[\vc\omega_n\times(\vt J\vc\omega_n)]\cdot\vc\psi_k\\
+2\int_{\R^3}\sym\nabla\vc u_n\cdot\nabla\vc\phi_k+\ell^2\int_{\R^3}\lap\vc u_n\cdot\lap\vc\phi_k-m_e\vc g\cdot\vc\eta_k+m_c(\vc r\times\vc g)\cdot\vc\psi_k\,,
\end{multline*}
where we denoted by $\vc c$ the $n$-dimensional vector of the coefficients of the truncated expansion which defines $\vc u_n$.

It is immediate to verify that $\vc\tau$ and $\mathcal T$ are continuous maps and that, for any $(\vc c,\vc g)\in B_R\times S^2$ we have $\vc\tau(\vc c,\vc g)\cdot\vc g=0$. Moreover,
\begin{multline*}
\mathcal T(\vc c,\vc g)\cdot\vc c=\rey\int_{\R^3}\left\{[(\vc u_n-\vc\xi_n-\vc\omega_n\times\vc x)\cdot\nabla]\vc u_n+\vc\omega_n\times\vc u_n\right\}\cdot\vc u_n\\ +\rey(m\vc\omega_n\times\vc\xi_n)\cdot\vc\xi_n+\rey[\vc\omega_n\times(\vt J\vc\omega_n)]\cdot\vc\omega_n\\
+2\int_{\R^3}\sym\nabla\vc u_n\cdot\nabla\vc u_n+\ell^2\int_{\R^3}\lap\vc u_n\cdot\lap\vc u_n-m_e\vc g\cdot\vc\xi_n+m_c(\vc r\times\vc g)\cdot\vc\omega_n\\
=2\int_{\R^3}\vert\sym\nabla\vc u_n\vert^2+\ell^2\int_{\R^3}\vert\lap\vc u_n\vert^2-m_e\vc g\cdot\vc\xi_n+m_c(\vc r\times\vc g)\cdot\vc\omega_n\,.
\end{multline*}
In the previous equality we used the fact that
\begin{equation*}
\int_{\R^3}[(\vc u_n-\vc\xi_n-\vc\omega_n\times\vc x)\cdot\nabla]\vc u_n\cdot\vc u_n=0\,,
\end{equation*}
which is easily proved, since $\dvg\vc u_n=0$, via tensorial identities and integration by parts. Now, by using Lemma \ref{lem:rigid}, there exists a constant $C>0$, independent of $n$, such that
\begin{multline}\label{eq:coercive}
\mathcal T(\vc c,\vc g)\cdot\vc c=\norm{\vc u_n}_Y^2-m_e\vc g\cdot\vc\xi_n+m_c(\vc r\times\vc g)\cdot\vc\omega_n\\
=R^2-(m_e+m_c)\vc g\cdot\vc\xi_n+m_c\vc g\cdot(\vc\xi_n+\vc\omega_n\times\vc r)\\
\geq R^2-m(|\vc\xi_n|+|\vc\xi_n+\vc\omega_n\times\vc r|)\geq R^2-CR>0\,,
\end{multline}
for $R$ large enough and
whenever $(\vc c,\vc g)\in\de B_R\times S^2$. Hence, Lemma \ref{lem:top} guarantees the existence, for any $n\in\N$, of a pair $(\bar{\vc c},\vc g_n)$ that produces a solution $(\vc u_n,\vc\xi_n,\vc\omega_n,\vc g_n)$ for the approximate problem given by
\begin{equation*}
\vc g_n\times\vc\omega_n=0\,,
\end{equation*}
and
\begin{multline*}
\rey\int_{\R^3}\left\{[(\vc u_n-\vc\xi_n-\vc\omega_n\times\vc x)\cdot\nabla]\vc u_n+\vc\omega_n\times\vc u_n\right\}\cdot\vc\phi_k\\ +\rey(m\vc\omega_n\times\vc\xi_n)\cdot\vc\eta_k+\rey[\vc\omega_n\times(\vt J\vc\omega_n)]\cdot\vc\psi_k\\
+2\int_{\R^3}\sym\nabla\vc u_n\cdot\nabla\vc\phi_k+\ell^2\int_{\R^3}\lap\vc u_n\cdot\lap\vc\phi_k-m_e\vc g_n\cdot\vc\eta_k+m_c(\vc r\times\vc g_n)\cdot\vc\psi_k=0\,,
\end{multline*}
for any $k\in\N$.

By taking into account \eqref{eq:coercive}, we see that the sequence $\{\vc u_n\}_{n\in\N}$ is bounded in $Y$, and hence weakly convergent to some $\vc u\in Y$; moreover, $\vc u_n\to\vc u$ strongly in $L^2_{loc}(\R^3;\R^3)$, and, clearly, $\vc\xi_n\to\vc\xi$, $\vc\omega_n\to\vc\omega$, and $\vc g_n\to\vc g$ in $\R^3$, so that 
%
%
the limits of the previous sequences are such that
\begin{equation*}
\vc g\times\vc\omega=0\,,
\end{equation*}
and
\begin{multline*}
\rey\int_{\R^3}\left\{[(\vc u-\vc\xi-\vc\omega\times\vc x)\cdot\nabla]\vc u+\vc\omega\times\vc u\right\}\cdot\vc\phi_k\\ +\rey(m\vc\omega\times\vc\xi)\cdot\vc\eta_k+\rey[\vc\omega\times(\vt J\vc\omega)]\cdot\vc\psi_k\\
+2\int_{\R^3}\sym\nabla\vc u\cdot\nabla\vc\phi_k+\ell^2\int_{\R^3}\lap\vc u\cdot\lap\vc\phi_k-m_e\vc g\cdot\vc\eta_k+m_c(\vc r\times\vc g)\cdot\vc\psi_k=0\,,
\end{multline*}
for any $k\in\N$.
Finally, by suitable linear combinations of the $\vc\phi_k$'s and the approximation property granted by Lemma \ref{lem:basis}, we can show that $(\vc u,\vc\xi,\vc\omega,\vc g)$ is a solution of equations \eqref{eq:nl7} and \eqref{eq:varform} for any $\vc v\in V$. The pressure field $p$ needed for the solution of the problem \eqref{eq:nl1}--\eqref{eq:nl7} can be recovered as the Lagrange multiplier of the divergence-free constraint imposed on the space $V$.
\end{proof}
\hyphenation{Na-zio-na-le}
\section*{Acknowledgments} 
This research is partially supported by GNFM (Gruppo Nazionale per la Fisica Matematica).


\end{document}